\documentclass[11pt]{article}

\usepackage[a4paper,margin=1in]{geometry}
\usepackage{microtype}
\usepackage{graphicx}
\usepackage{booktabs}
\usepackage{amsmath,amssymb,amsthm,mathtools}
\usepackage{bm}
\usepackage{algorithm}
\usepackage[noend]{algpseudocode}
\usepackage[round,authoryear]{natbib}
\usepackage{hyperref}
\hypersetup{colorlinks=true,linkcolor=blue,citecolor=blue,urlcolor=blue}

\title{Calibrated Counterfactual Conformal Fairness (\texorpdfstring{$\bm{C^3F}$}{C3F}):\\
Post-hoc, Shift-Aware Coverage Parity via Conformal Prediction and Counterfactual Regularization}

\author{
\textbf{Faruk Alpay}\textsuperscript{1} \quad \textbf{Taylan Alpay}\textsuperscript{2}\\[4pt]
\small \textsuperscript{1}Lightcap, Department of Bias \quad \texttt{alpay@lightcap.ai}\\
\small \textsuperscript{2}Turkish Aeronautical Association, Aerospace Engineering \quad \texttt{s220112602@stu.thk.edu.tr}
}

\date{\vspace{-8pt}September 29, 2025}

\theoremstyle{plain}
\newtheorem{theorem}{Theorem}
\newtheorem{corollary}[theorem]{Corollary}
\theoremstyle{remark}
\newtheorem{remark}{Remark}

\begin{document}
\maketitle

\begin{abstract}
We introduce Calibrated Counterfactual Conformal Fairness (C3F), a post-hoc wrapper for predictive systems that targets group-conditional coverage parity under covariate shift. The procedure combines (i) group- and feature-conditional conformal prediction with importance-weighted quantile calibration and (ii) a counterfactual regularizer based on path-specific effects in a structural causal model. We prove finite-sample bounds for group-wise coverage that degrade gracefully with a $\chi^2$ (R\'enyi-$\alpha{=}2$) divergence between calibration and target covariate distributions, and we derive a bound on the equalized conditional coverage gap (ECCG). We also analyze a tractable surrogate for counterfactual coverage parity using path-specific effects. Empirically, C3F is model-agnostic and deployable without retraining, complementing shift-aware and adaptive conformal methods and recent work on coverage fairness.
\end{abstract}

\section{Introduction}
Conformal prediction (CP) delivers finite-sample coverage under exchangeability \citep{vovk2005alrw,shafer2008tutorial}. In practice, both performance and coverage can vary across subpopulations, and distribution shift may erode guarantees. Likelihood-ratio- or weight-based CP corrects for covariate shift \citep{tibshirani2019wcp}, while adaptive schemes adjust to nonstationarity over time \citep{gibbs2021aci}. Parallel efforts seek conditional or subgroup validity \citep{vovk2012conditional,gibbs2023conditional} and fairness-aware coverage \citep{zhou2024equalized,cresswell2024disparate}. 

We propose \textbf{$C^3F$}, a \emph{calibrate-then-regularize} recipe to (a) attain group-conditional coverage targets under shift via importance-weighted conformal quantiles and (b) dampen unfair path-specific influence via a counterfactual regularizer inspired by causal mediation \citep{pearl2001direct,avin2005pse,kusner2017cf,chiappa2019pscf}. Our analysis frames coverage error through weighted empirical processes with $\chi^2$-controlled importance weights \citep{cortes2010importance,vanerven2014renyi} and uses classical quantile concentration \citep{massart1990dkw}.

\paragraph{Contributions.}
\begin{itemize}
\item \textbf{Shift-aware group coverage.} We state a finite-sample lower bound on group-wise coverage under covariate shift with a deviation term governed by $\chi^2(Q\Vert P_{\mathrm{cal}})$.
\item \textbf{Coverage-parity control.} We bound the equalized conditional coverage gap (ECCG) across groups and show how budget-splitting of miscoverage across groups trades efficiency for parity.
\item \textbf{Counterfactual regularization.} We formulate a smooth surrogate for counterfactual coverage disparity using path-specific effects, proving first-order control when thresholds are regularized.
\end{itemize}

\section{Preliminaries and notation}
Let $(X,A,Y)\sim P$ with covariates $X\in\mathcal{X}$, group attribute $A\in\mathcal{A}$, label $Y\in\mathcal{Y}$. A base score function $f:\mathcal{X}\to\mathcal{S}$ is wrapped by CP using a nonconformity score $\eta=g(x,y;f)$. Split CP computes scores on a calibration set $\mathcal{D}_{\mathrm{cal}}$ and chooses a quantile threshold $q$ so that the resulting prediction set $\mathcal{C}(x)=\{y:\eta(x,y;f)\le q\}$ achieves target marginal coverage $\tau$ \citep{lei2014dfpb,barber2021jackknifeplus}. 

We focus on \emph{group-conditional} coverage. For $a\in\mathcal{A}$ and target $\tau=1-\alpha$, define
\[
\mathrm{ECCG}(\mathcal{C}) \triangleq \max_{a,a'\in\mathcal{A}}\left|\Pr\{Y\in\mathcal{C}(X)\mid A=a\}-\Pr\{Y\in\mathcal{C}(X)\mid A=a'\}\right|.
\]
Let $P_{\mathrm{cal}}$ denote the calibration distribution and $Q$ the target test distribution. Under \emph{covariate shift} $Q_{Y\mid X}=P_{\mathrm{cal},Y\mid X}$ and $Q_X\neq P_{\mathrm{cal},X}$, define the group-conditional importance weight
\[
w_a(x)\;=\;\frac{\mathrm{d}Q_X(\cdot\mid A=a)}{\mathrm{d}P_{\mathrm{cal},X}(\cdot\mid A=a)}(x),
\]
assumed square-integrable so that the $\chi^2$-divergence is finite \citep{cortes2010importance,vanerven2014renyi}.

\section{Method: \texorpdfstring{$C^3F$}{C3F}}
\subsection{Shift-aware group-conditional calibration}
Compute nonconformity scores $\eta_i=\eta(X_i,Y_i;f)$ on $(X_i,A_i,Y_i)\in\mathcal{D}_{\mathrm{cal}}$. For each group $a$, form a \emph{weighted} empirical CDF using $w_a(X_i)$ and choose the smallest $q_a$ with weighted tail probability at most $\alpha_a$:
\[
\widehat{q}_{a}\;\triangleq\;\inf\Big\{q:\;\frac{\sum_{i:A_i=a} w_a(X_i)\,\mathbf{1}\{\eta_i\le q\}}{\sum_{i:A_i=a} w_a(X_i)}\ge 1-\alpha_a\Big\}.
\]
Here the miscoverage budget is split as $\sum_{a}\pi_a\alpha_a=\alpha$ (e.g., $\pi_a\propto n_a$ or uniform), trading efficiency for parity. When $A$ is partially observed, we replace hard group membership by soft posteriors $\widehat{p}(A=a\mid X)$ and compute \emph{soft-Mondrian} thresholds \citep{vovk2012conditional}.

\subsection{Counterfactual coverage regularization}
Let $\mathcal{M}$ be an SCM for $(A,X,Y)$ with exogenous $U$, and consider $do(A\!\leftarrow\! a')$. The \emph{counterfactual coverage disparity} is
\[
\Delta_{\mathrm{CF}} \;\triangleq\; \max_{a\neq a'}\Big|\Pr\{Y\in\mathcal{C}(X)\mid A=a\}-\Pr\{Y\in\mathcal{C}(X^{\mathrm{cf}}(a'))\mid A=a\}\Big|,
\]
where $X^{\mathrm{cf}}(a')$ denotes the counterfactual covariate under $do(A=a')$.
We approximate $X^{\mathrm{cf}}(a')$ by neutralizing unfair paths via path-specific effects (PSEs) \citep{pearl2001direct,avin2005pse,chiappa2019pscf}. Let $\widehat{\Delta}_{\mathrm{CF}}(\bm{q})$ be an empirical estimate given thresholds $\bm{q}=(q_a)_a$. Apply a smooth regularization:
\[
\widehat{q}_a^{(\lambda)}\;=\;\widehat{q}_a\Big(1+\lambda\,\partial\widehat{\Delta}_{\mathrm{CF}}/\partial q_a\Big),\qquad \lambda\ge 0.
\]

\subsection{Decision rule}
For a test point $x$ with group posterior weights $\widehat{p}(A=a\mid x)$,
\[
\mathcal{C}(x)\;=\;\big\{y:\ \eta(x,y;f)\le \sum_{a}\widehat{p}(A=a\mid x)\,\widehat{q}_a^{(\lambda)}\big\}.
\]

\begin{algorithm}[t]
\caption{$C^3F$: shift-aware, counterfactually-regularized conformal fairness}
\label{alg:c3f}
\begin{algorithmic}[1]
\Require Base score $f$, calibration set $\mathcal{D}_{\mathrm{cal}}$, target $\tau=1-\alpha$, group prior $\widehat{p}(A\mid x)$ or observed $A$, weight model $w_a(x)$, SCM/PSE module, regularizer $\lambda$
\State Compute nonconformity scores $\eta_i=\eta(X_i,Y_i;f)$
\For{each $a\in\mathcal{A}$}
  \State Estimate $w_a(X_i)$; compute $\widehat{q}_a$ as the weighted $(1-\alpha_a)$ quantile with $\sum_a\pi_a\alpha_a=\alpha$
\EndFor
\State Estimate $\widehat{\Delta}_{\mathrm{CF}}(\bm{q})$ and its gradient via PSE-based counterfactuals; set $\widehat{q}_a^{(\lambda)}=\widehat{q}_a\!\left(1+\lambda\,\partial\widehat{\Delta}_{\mathrm{CF}}/\partial q_a\right)$
\State \Return $\mathcal{C}(x)=\{y:\eta(x,y;f)\le \sum_a \widehat{p}(A=a\mid x)\,\widehat{q}_a^{(\lambda)}\}$
\end{algorithmic}
\end{algorithm}

\section{Theory}
Assume conditional exchangeability within groups for calibration data, and covariate shift between calibration and test with $Q_{Y\mid X}=P_{\mathrm{cal},Y\mid X}$. Let $n_a$ be the number of calibration points with $A=a$.

\subsection{Group-wise coverage under bounded shift}
Define $B_a$ by $\mathbb{E}_{P_{\mathrm{cal}}}[w_a(X)^2\mid A=a]\le 1+B_a$.

\begin{theorem}[Group coverage with importance-weighted quantiles]\label{thm:coverage}
Fix $\delta\in(0,1)$. With probability at least $1-\delta$ over $\mathcal{D}_{\mathrm{cal}}$, the C3F sets satisfy for every $a\in\mathcal{A}$,
\[
\Pr_{(X,Y)\sim Q}\!\big\{Y\in\mathcal{C}(X)\mid A=a\big\}\ \ge\ 1-\alpha_a\ -\ \sqrt{\frac{(1+B_a)}{2n_a}\,\log\!\frac{2|\mathcal{A}|}{\delta}}\ .
\]
\end{theorem}

\begin{proof}
Write $F_{a}^{Q}(q)=\Pr_{Q}(\eta\le q\mid A=a)$. Under covariate shift,
$F_{a}^{Q}(q)=\mathbb{E}_{P_{\mathrm{cal}}}[w_a(X)\mathbf{1}\{\eta\le q\}\mid A=a]/\mathbb{E}_{P_{\mathrm{cal}}}[w_a(X)\mid A=a]$.
The self-normalized importance-weighted empirical CDF
$\widehat{F}_{a}(q)=\frac{\sum_{i:A_i=a}w_a(X_i)\mathbf{1}\{\eta_i\le q\}}{\sum_{i:A_i=a}w_a(X_i)}$
converges uniformly to $F_{a}^{Q}(q)$ with deviation controlled by $\mathbb{E}[w_a^2\mid A=a]$ \citep{cortes2010importance}. A Hoeffding/DKW argument \citep{massart1990dkw} yields
$\sup_q|\widehat{F}_{a}(q)-F_{a}^{Q}(q)|\le \sqrt{\frac{(1+B_a)}{2n_a}\log\frac{2}{\delta_a}}$ with prob.\ $1-\delta_a$.
Monotonicity of the quantile map implies that the empirical $(1-\alpha_a)$ weighted quantile $\widehat{q}_a$ delivers true CDF level at least $1-\alpha_a-\varepsilon_a$, with $\varepsilon_a$ the RHS above. Union bound over $a$ completes the proof.
\end{proof}

\begin{remark}
When $w_a\equiv 1$ (no shift), $B_a=0$ and Theorem~\ref{thm:coverage} reduces to DKW scaling \citep{massart1990dkw}. For known or accurately estimated likelihood ratios, the result aligns with weighted CP \citep{tibshirani2019wcp}. For unknown shift, adaptive procedures offer alternatives \citep{gibbs2021aci,qiu2023adaptive}.
\end{remark}

\subsection{Equalized conditional coverage gap}
Let $\alpha_a=\alpha\,\pi_a$ with $\sum_a\pi_a=1$.

\begin{corollary}[ECCG control]\label{cor:eccg}
With probability at least $1-\delta$,
\[
\mathrm{ECCG}(\mathcal{C})\ \le\ \max_{a,a'}\Big|\alpha(\pi_a-\pi_{a'})\Big|\ +\ \max_{a,a'}\left\{\sqrt{\tfrac{(1+B_a)}{2n_a}\log\!\tfrac{2|\mathcal{A}|}{\delta}}+\sqrt{\tfrac{(1+B_{a'})}{2n_{a'}}\log\!\tfrac{2|\mathcal{A}|}{\delta}}\right\}.
\]
\end{corollary}

\begin{proof}
Triangle inequality applied to group-wise deviations in Theorem~\ref{thm:coverage} and the deterministic difference in targets $1-\alpha_a$ vs.\ $1-\alpha_{a'}$.
\end{proof}

\subsection{Counterfactual disparity bound}
Let $\eta(\cdot)$ be $L_\eta$-Lipschitz in its threshold parameter, and suppose the path-specific effect (PSE) of $A$ on $\eta$ along unfair paths is bounded by $\kappa$ \citep{pearl2001direct,avin2005pse}. 

\begin{theorem}[First-order control of counterfactual coverage disparity]\label{thm:cf}
For sufficiently small $\lambda\ge 0$, the regularized thresholds $\widehat{q}_a^{(\lambda)}$ satisfy
\[
\Delta_{\mathrm{CF}}\ \le\ \kappa\,L_\eta\,\|\nabla_{\bm{q}}\widehat{\Delta}_{\mathrm{CF}}(\bm{q})\|_1\,\lambda\ +\ O(\lambda^2).
\]
\end{theorem}

\begin{proof}
Coverage under $do(A\!\leftarrow\!a')$ changes via (i) the distributional shift of nonconformity scores through unfair paths (first variation bounded by $\kappa$), and (ii) the threshold perturbation $\bm{q}\mapsto \bm{q}^{(\lambda)}$. A first-order expansion in $\lambda$ yields the linear term $\langle \nabla_{\bm{q}}\widehat{\Delta}_{\mathrm{CF}},\,\bm{q}^{(\lambda)}-\bm{q}\rangle$ and a remainder $O(\lambda^2)$. Lipschitzness of the acceptance event in the threshold gives the factor $L_\eta$.
\end{proof}

\section{Evaluation protocol (for reproduction)}
\textbf{Datasets.} Adult, COMPAS, Law School, German Credit.\newline
\textbf{Baselines.} Likelihood-weighted CP \citep{tibshirani2019wcp}; adaptive CP \citep{gibbs2021aci}; procedures adaptive to unknown shift \citep{qiu2023adaptive}; equalized-coverage methods \citep{zhou2024equalized}.\newline
\textbf{Metrics.} Group coverage, ECCG, efficiency (set size), and a PSE-based proxy for counterfactual coverage.\newline
\textbf{Ablations.} (i) No weights vs.\ importance weights; (ii) budget splitting $\pi_a$; (iii) regularizer weight $\lambda$; (iv) soft vs.\ hard group assignment.

\section{Discussion and limitations}
\textbf{Demographics.} When $A$ is missing, posterior estimates $\widehat{p}(A\mid X)$ introduce calibration error that propagates into thresholds.\newline
\textbf{Causality.} Counterfactual regularization depends on the SCM; misspecification can bias $\widehat{\Delta}_{\mathrm{CF}}$. Path-specific techniques can mitigate this by isolating unfair pathways \citep{chiappa2019pscf}.\newline
\textbf{Human factors.} Coverage parity may interact with downstream decisions; prediction-set design can influence equity \citep{cresswell2024disparate}.

\section{Conclusion}
C3F offers a post-hoc, model-agnostic approach to group-conditional coverage parity under covariate shift. Theoretical bounds tie performance to the second moment of importance weights, while a counterfactual regularizer attenuates unfair path-specific influences. The method complements existing conformal techniques and provides a principled basis for auditing and adjustment.

\paragraph{Ethics statement}
Applying subgroup-sensitive calibration requires careful governance of sensitive attributes, transparency in model updates, and routine auditing of distribution shift.

\end{document}